\documentclass[]{IEEEtran}
\makeatletter
\def\ps@headings{%
\def\@oddhead{\mbox{}\scriptsize\rightmark \hfil \thepage}%
\def\@evenhead{\scriptsize\thepage \hfil \leftmark\mbox{}}%
\def\@oddfoot{}%
\def\@evenfoot{}}
\makeatother
\usepackage{cite,url}
\usepackage{booktabs}
\usepackage{amsmath,amssymb,bm}
\usepackage[dvips]{color}
\usepackage{graphicx}
\usepackage[lined,boxed,commentsnumbered, ruled]{algorithm2e}
\usepackage{booktabs}
\usepackage{subfigure}
\usepackage{amsfonts}
\allowdisplaybreaks

\IEEEaftertitletext{\vspace{-5.5ex}}%

\newcommand {\dd}{\mathrm{d}}
\newcommand {\bx}{\mathbf{x}}
\newcommand {\by}{\mathbf{y}}

\newcommand {\bzero}{\mathbf{0}}

\newcommand {\bB}{\mathbf{B}}

\newcommand {\bBs}{\mathbf{B}^{*}}
\newcommand {\bBsb}{\widetilde{\mathbf{B}}^{*}}
\newcommand {\bBsc}{\widehat{\mathbf{B}}^{*}}

\newcommand {\bc}{\mathbf{c}}
\newcommand {\bcs}{\mathbf{c}^{*}}
\newcommand {\bcsb}{\widetilde{\mathbf{c}}^{*}}
\newcommand {\bcsc}{\widehat{\mathbf{c}}^{*}}

\newcommand {\As}{A^{*}}
\newcommand {\Asb}{\widetilde{A}^{*}}
\newcommand {\Asc}{\widehat{A}^{*}}
\newcommand {\Ass}{A^{**}}

\newcommand {\etas}{\eta^{*}}
\newcommand {\etasb}{\widetilde{\eta}^{*}}
\newcommand {\etasc}{\widehat{\eta}^{*}}
\newcommand {\etass}{\eta^{**}}

\newcommand {\bAs}{\mathbf{A}^{*}}
\newcommand {\bAss}{\mathbf{A}^{**}}
\newcommand {\bAsb}{\widetilde{\mathbf{A}}^{*}}

\newcommand {\bAsd}{\mathbf{A}'^{*}}

\newcommand {\Asd}{A'^{*}}

\newcommand {\bA}{\mathbf{A}}
\newcommand {\bfeta}{\boldsymbol{\eta}}
\newcommand {\betas}{\boldsymbol{\eta}^{*}}
\newcommand {\betass}{\boldsymbol{\eta}^{**}}
\newcommand {\betasb}{\widetilde{\boldsymbol{\eta}}^{*}}
\newcommand {\betasc}{\widehat{\boldsymbol{\eta}}^{*}}
\newcommand {\Mb}{\widetilde{M}}

\newtheorem{theorem}{Theorem}
\newtheorem{lemma}{Lemma}

\begin{document}
\title{\fontsize{23pt}{25pt}\selectfont{Structured Spectrum Allocation and User Association in Heterogeneous Cellular Networks}}
\author{Wei Bao and Ben Liang  \\
Department of Electrical and Computer Engineering, University of Toronto,  Canada\\
              Email: \{wbao, liang\}@comm.utoronto.ca
}

\maketitle
\thispagestyle{empty}

\begin{abstract}
We study  joint spectrum allocation  and   user association  in
heterogeneous cellular networks with multiple tiers  of base stations. A stochastic geometric approach is applied as the basis to derive  the average downlink user data rate  in a closed-form expression. Then, the expression is employed as the objective function in jointly optimizing spectrum allocation and   user association, which is of non-convex programming in nature. A  computationally efficient Structured Spectrum Allocation and User Association (SSAUA) approach is proposed, solving the optimization problem optimally when the density of users is low, and near-optimally with a guaranteed performance bound when the  density of users is high. A Surcharge Pricing Scheme (SPS) is also presented,  such that the designed association bias values can be achieved in Nash equilibrium. Simulations and numerical studies are conducted to validate the accuracy and efficiency of the proposed SSAUA approach and SPS.
\end{abstract}
\section{Introduction}\label{section_intro}

Traditional single-tiered macro-cellular networks provide wide coverage for mobile user equipments (UEs), but they are insufficient to satisfy the exploding  demand driven by modern mobile traffic, such as multimedia transmissions and cloud computing tasks.  One efficient means to alleviate this problem is to install a diverse set of small-cells (e.g., picocells and femtocells), overlaying the macrocells, to form a multi-tiered  heterogeneous cellular network \cite{Survey}. Each small-cell is equipped with a shorter-range and lower-cost base station (BS), to provide  nearby UEs with higher-quality communication links with lower power usage.

However, in the presence of multiple tiers of BSs in a cellular network, user association control becomes more challenging. A most direct approach is association by maximum received power, in which UEs are associated with the BS (in any tier) with the largest received power. However, in this case, because small-cell BSs are transmitting at lower power levels, only the UEs very close to them will connect with them, while most other UEs are still crowding in macrocells, leading to degraded performance. An example is shown in Fig.~\ref{fig_intro1}, in which many UEs are occupying the macrocells, while some small-cells are nearly empty.

In order to resolve this issue, a \textit{flexible user association} approach may be employed \cite{JF1,JF3}, in which
 each tier of BSs is assigned  a \textit{user association bias value}, and a UE is associated with a BS with the largest received power multiplied by the bias value. If small-cell BSs are assigned with  larger association bias values, the small-cells are ``expanded'' accordingly.
This can result in a more balanced mobile traffic pattern and thus better network performance. Fig.~\ref{fig_intro2} shows an example of flexible user association.
 However, if the association bias values for small-cell BSs are too large, it will cause improper  expansions of small-cells such that UEs at their cell-edge may suffer from inadequate received power. As a consequence, the association bias values should be properly designed so that the overall network performance is optimized.

%Although some efforts have been made to  provide analytical framework to study the network performance given the bias values. It is still debatable how to derive the optimal connecting bias values with arbitrary number of tiers .

Further complicating the resource management problem in a multi-tier cellular network, the radio spectrum licensed by the network operator needs to be shared by BSs of widely different power and coverage areas. How to optimally allocate spectrum among different tiers is an important open problem. In order to avoid cross-tier interference, and the prohibitive complexity in tracking and provisioning for such interference especially with unplanned deployment of small cells, a disjoint spectrum mode is commonly advocated \cite{SG_ICASSP, SpectrumAllocation, TQ1}, where different tiers of BSs are allocated non-overlapping portions of the spectrum. Even so, it is still a challenging problem to properly divide the spectrum for optimal network performance.

\begin{figure} \centering
\addtolength{\subfigcapskip}{-0.3cm}
\subfigure[Association by maximum received power.] { \label{fig_intro1}
\includegraphics[scale=0.45]{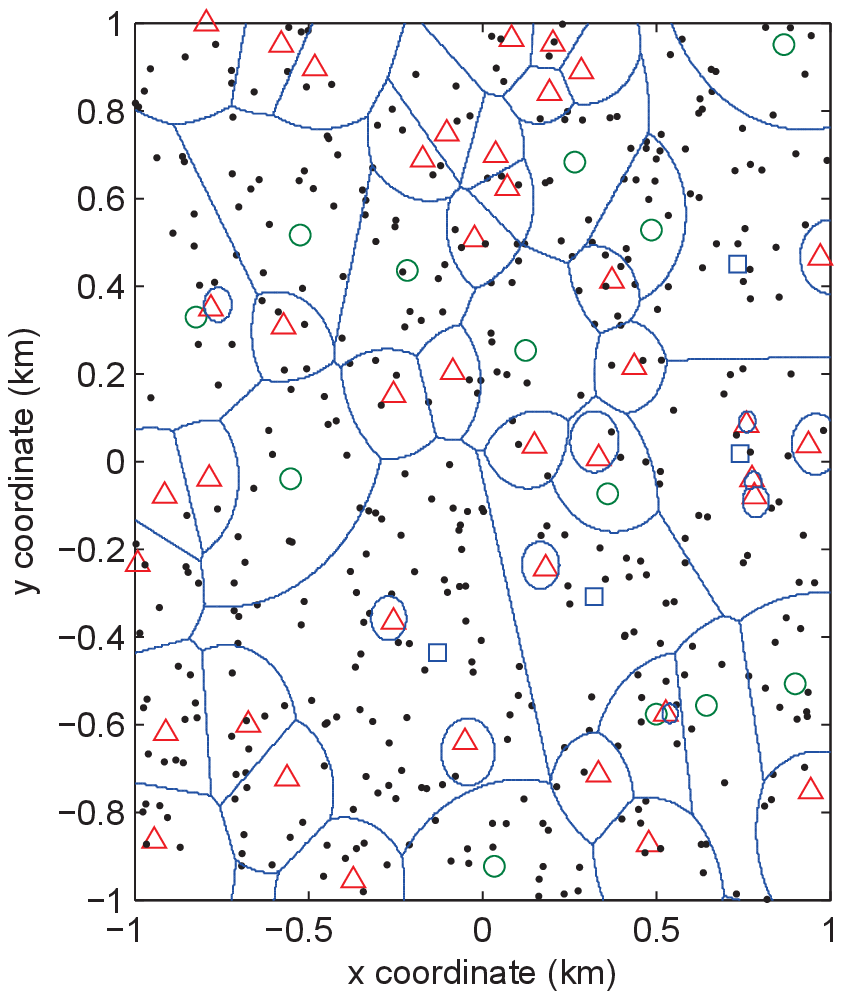}
}
\subfigure[Flexible user association.] { \label{fig_intro2}
\includegraphics[scale=0.45]{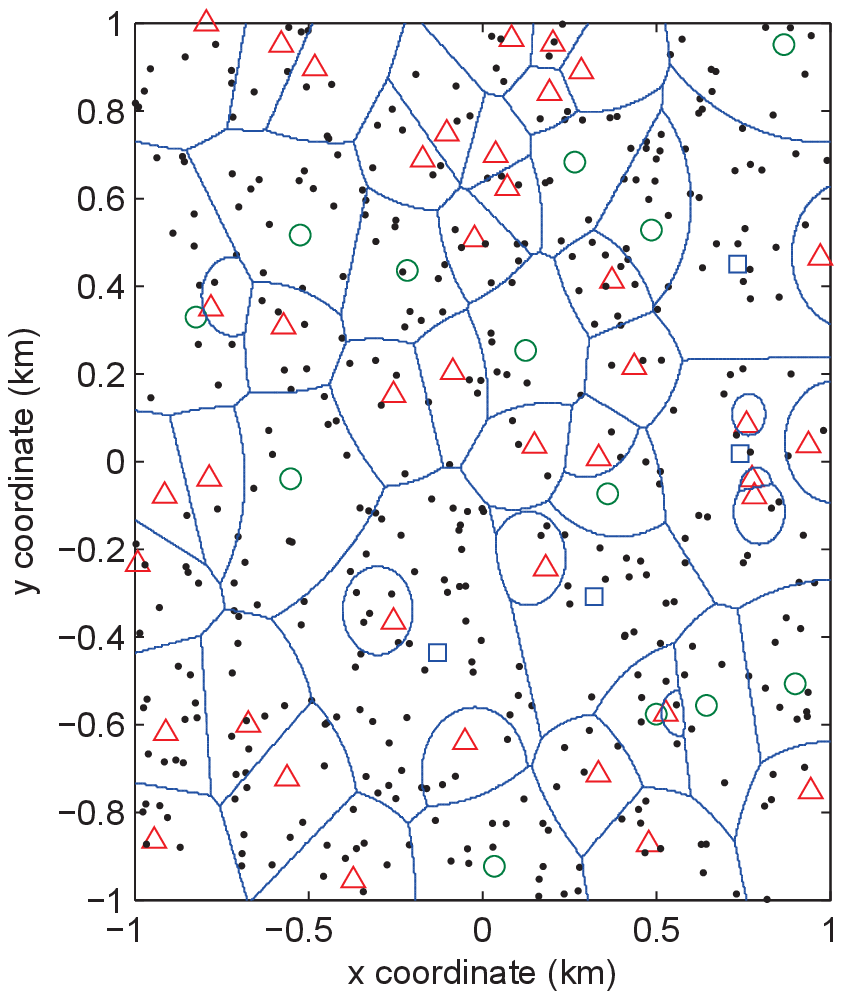}
}
\vspace{-0.3 cm}
\caption{An example of a three-tier cellular network. Macrocell BSs, picocell BSs, and femtocell BSs are represented by squares, circles, and triangles respectively; UEs are represented by dots; blue lines show cell boundaries.}
\label{fig_intro}
\end{figure}

In this work, our objective is to study jointly optimal spectrum allocation and user association in a heterogeneous celluar  network with  multiple tiers of BSs.
First, we develop a stochastic geometric model to study the network performance analytically.  %Each tier of BSs is modeled as a Poisson point process (PPP), which captures their random spatial patterns.
A closed-form expression for the average downlink UE data rate is derived, which is then employed as the objective function for jointly optimizing the spectrum allocation among tiers and the user association bias values.

This resultant optimization problem is of non-convex programming in nature and cannot be solved with a standard method. Instead, we explore two important structures in solving the problem. Referred to as the \textit{density thresholding} structure, we show that the problem can be studied separately over  sparse UE and dense UE scenarios, divided by a parameter specific threshold. Referred to as the \textit{priority ordering} structure, we show that a tier with higher BS density should have higher priority in spectrum allocation. Based on these observations, we propose a computationally efficient Structured Spectrum Allocation and User Association (SSAUA) approach to solve the problem optimally in the sparse UE scenario, and near-optimally in the dense UE scenario with a quantified performance bound.

%In the sparse UE scenario, after characterizing some properties of the optimal solutions, we solve the optimization problem through a structured method without computational complexity. In the dense UE scenario, we use an approximated objective function to replace the original objective function, in order to solve the approximated one through a structured method. A comparison between the approximated objective function and the original one is made, deriving the performance gap

Finally, toward practical implementation of SSAUA, we propose a Surcharge Pricing Scheme (SPS), such that the designed association bias values can be achieved in Nash equilibrium. Hence, each UE is incentivized to adopt the proposed design with individual rationality.

The rest of the paper is organized as follows. In section \ref{section_related}, we discuss the relation between our work and prior works.
In Section \ref{section_model}, we present the system model. In Sections \ref{section_performance}, \ref{section_optimization}, and \ref{section_equilibirum}, we present our contributions in UE data rate derivation, SSAUA design, and SPS, respectively. In Section \ref{section_experiment}, we present numerical results. Finally, conclusions are given in Section \ref{section_conclusion}.

\section{Related Works}\label{section_related}

\subsection{Stochastic Geometry as Analytical Basis}
Stochastic geometry  \cite{SG_Totorial1,SG_Totorial2, SG_Totorial3, SG_book1} is a powerful mathematical modeling tool to
analyze the performance (e.g., outage probability and data rate) of cellular networks with random spatial patterns of UEs and BSs. In this work, we focus on the downlink user data rate as performance measure.
Pioneering works on downlink performance analysis using stochastic geometry include \cite{JF4}, \cite{JF5}, and \cite{arXiv_downlink}, for either the single-tier or the multi-tier case. None of them considered spectrum allocation or user association.

\subsection{Spectrum Allocation and User Association}
Most prior studies considered either spectrum allocation or user association separately. For example, \cite{TQ1,TQ2} studied optimal spectrum allocation, in cellular networks limited to two tiers of BSs, without flexible user association.

Assuming a fixed number of UEs and BSs and without considering their random spatial patterns, \cite{noSGShroff, noSGJF, noSGMHong, noSGFB, noSGWeiYu_ICASSP} studied optimal user association with deterministic utility optimization. With a stochastic geometric approach,  \cite{JF1} proposed the flexible user association model with bias values, which is adopted in our work. It also derived the coverage probability and UE data rate, considering cross-tier interference, but in non-closed forms.  \cite{JF1} did not provide any design insights to optimize the derived performance metrics. It was later extended in \cite{JF3} to study optimal user association in a network with two tiers of BSs, without considering spectrum allocation.

Joint spectrum allocation and user association was studied in \cite{JF2}. It was limited to two tiers of BSs. \cite{JF2} presented a qualitative study on the optimal network performance in terms of coverage probability and data rate, without further providing analytical details in solving the optimization problems. A similar problem was also studied in  \cite{Yicheng}, with frequency reuse instead of tiered spectrum division as the approach for spectrum sharing. It could accommodate more than two tiers, but it provided only conditionally optimal user association given frequency reuse factors or conditionally optimal frequency reuse factors given user association bias values. Joint optimization remained an open problem.

Compared with the above studies, we consider multiple tiers of BSs with disjoint spectrum and provide optimal and analytically bounded near-optimal solutions for joint spectrum allocation and user association.

\section{System Model}\label{section_model}

\subsection{Multi-tier Cellular Network}
We consider a heterogeneous cellular network with randomly spatially distributed $K\geq 2$ tiers of BSs.
As in conventional stochastic geometric modeling of multi-tier cellular networks \cite{JF1,JF2,JF3,JF5,Yicheng}, each tier of BSs independently forms a homogeneous Poisson point process (PPP) in two dimensional Euclidean space $\mathbb{R}^2$. Let $\Phi_k$ denote the PPP corresponding to tier-$k$ BSs, with intensity $\lambda_k$. Without loss of generality, we assume that $\lambda_1<\lambda_2\ldots<\lambda_K$. (If $\lambda_i=\lambda_j, i\neq j$ in reality, we may approximate $\lambda_j=\lambda_i+\xi$, where $\xi$ is arbitrarily close to $0$.)  UEs are also modeled as a homogeneous PPP $\Psi$ with intensity $\mu$, independent of all BSs. We assume each BS is connected to the core network by separate high-capacity wired or wireless links that have no influence on our performance analysis. In addition, because we focus on downlink  analysis, we assume that the downlink and uplink of the system are
operated in different spectrum, so that the uplink interference and capacity have no influence on the downlink analysis.

%$\bfeta$ will be used as a decision variable, to optimize the system performance in Section \ref{section_optimization}.

\subsection{Power and Pathloss Model}
We define the tiers of BSs by their transmission power. Let  $P_k$ be the transmission power of  tier-$k$ BSs,  which is a given parameter.
If $P_t(\bx)$, $P_t(\bx)\in \{P_1, P_2,\ldots, P_K\}$, is the transmission power from a BS at $\bx$ and $P_r(\by)$ is the received power at $\by$, we have $P_r(\by)=\frac{P_t(\bx)h_{\bx,\by}}{\alpha|\bx-\by|^{\gamma}}$, where $\alpha|\bx-\by|^{\gamma}$ is the propagation loss function with predetermined constants $\alpha$ and $\gamma$ (where the path loss exponent $\gamma>2$ in practice),  and $h_{\bx,\by}$ is the fast fading term.  Corresponding to common Rayleigh fading with power normalization, $h_{\bx,\by}$ is independently  exponentially distributed with unit mean. Let $H(\cdot)$ be the cumulative distribution function  of $h_{\bx,\by}$.

\subsection{Spectrum Allocation} \label{subsection_allocation}
In order to avoid cross-tier interference, different tiers of BSs are allocated separated spectrum. Assume the total spectrum bandwidth is $W$.  The network operator allocates $\eta_kW$ to each tier-$k$ BS, where $\eta_k$ is the spectrum allocation factor and $\sum_{k=1}^{K}\eta_k=1$. Let $\bfeta=(\eta_1,\eta_2,\ldots, \eta_K)$. Note that  BSs in the same tier are operated on the same spectrum.

We additionally consider the possible constraints $\eta_{\min,k} \leq \eta_k \leq \eta_{\max,k}$, for $k=1, 2, \ldots, K$.  Clearly, we have $\sum_{k=1}^{K} \eta_{\min,k} \leq 1\leq \sum_{k=1}^{K} \eta_{\max,k}$.  Furthermore, we assume that $0<\eta_{\min,1}\leq \eta_{\min,2}\ldots\leq \eta_{\min,K}$ and $0<\eta_{\max,1}\leq\eta_{\max,2}\ldots\leq\eta_{\max,K}$, i.e., the network operator is likely (but not necessarily) to allocate more spectrum to a tier with higher BS density.  Note that this is a general condition that
contains the special case where there is no constraint on $\bfeta$.

Given a specific tier-$k$ BS, it is common to assume that all its associated UEs are equally allocated spectrum \cite{Yicheng,JF1,TQ1}. Hence, the per-UE assigned spectrum bandwidth is $\beta_k=\eta_k W/N_k$, where $N_k$ is a random variable denoting the number of UEs associated with the BS.

\subsection{Coverage Probability and UE Data Rate} \label{subsection_datarate}
Let $T$ denote the minimum required Signal-to-Interference Ratio (SIR) of UEs.  The coverage probability of  a UE is defined as the probability that its SIR  is no less than $T$ \cite{SG_Totorial1}. As in conventional wireless modeling \cite{Yicheng,JF1,TQ1}, if a UE experiences coverage probability $P'$ and is allocated spectrum bandwidth $\beta'$, its data rate is $\beta' \log(1+T)$ if the SIR is no less than $T$, and its data rate is $0$ if the SIR is  less than $T$ (i.e., outage occurs). Thus, the overall data rate of the UE is $\beta' \log(1+T) P'$. Note that,  unless otherwise stated, $\log$ represents $\log_2$. Also, the thermal noise is assumed to be negligible compared with interference.

\subsection{Flexible User Association}
Given that a UE is located at $\by$, it associates itself with the BS that provides the maximum \textit{biased received power}  \cite{JF1,JF3,Yicheng} as follows:
\begin{align}\label{formula_def1}
\mathcal{BS}(\by)=\arg \max_{\bx\in \Phi_k, \forall k}B_kP_k|\bx-\by|^{-\gamma},
\end{align}
where $\mathcal{BS}(\by)$ denotes the location of the BS associated with the UE, and $P_k|\bx-\by|^{-\gamma}$ is the received power from a tier-$k$ BS located at $\bx$, and $B_k$ is the association bias, indicating the connecting preference of a UE toward tier-$k$ BSs. In this case, the resultant cell splitting forms a generalized  Dirichlet tessellation, or weighted Poisson Voronoi \cite{PV}, shown in Fig.~\ref{fig_intro2}.
 %Thus (\ref{formula_def1}) shows that each UE is connecting to the BS with the maximum biased received power.
Note that for $B_1, B_2,\ldots, B_K$, their effects remain the same if we multiply all of them by the same positive constant. Thus, without loss of generality, we normalize them such that $\sum_{k=1}^{K}B_k=1$. Let $\bB=(B_1,B_2,\ldots, B_K)$.

Let $A_k$ denote the probability that a UE associates itself with a tier-$k$ BS, and $\bA=(A_1,A_2,\ldots,A_K)$.
As derived in \cite{JF1}, we have
%$\forall j$, we first derive the pdf of the distance from the typical UE to its nearest type $j$ BS. Then, we can derive $A_k$ through multi-level integrals over the pdfs. See [][][] for more details.
\begin{align}\label{formula_Ak}
A_k=\frac{\lambda_k(P_kB_k)^{\frac{2}{\gamma}}}{\sum_{j=1}^{K}\lambda_j(P_jB_j)^{\frac{2}{\gamma}}},
\end{align}
and thus%in which $\sum_{k=1}^{K}A_k=1$.  Let $\bA=(A_1,A_2,\ldots,A_K)$.  Also, as we set   $\sum_{k=1}^{K}B_k=1$,
\begin{align}\label{formula_Bk}
B_k=\frac{P_k^{-1}(A_k/\lambda_k)^{\frac{\gamma}{2}}}{\sum_{j=1}^{K}P_j^{-1}(A_j/\lambda_j)^{\frac{\gamma}{2}}}.
\end{align}
Hence, there is a one-to-one mapping between $\bA$ and $\bB$, so we can view them  interchangeably.

\subsection{Problem Statement}

We first  aim to derive a closed-form expression for the average UE data rate.  Then, our objective is to maximize the average UE data rate by jointly optimizing the spectrum allocation factors $\bfeta$ and the user association bias values $\bB$ (or equivalently $\bA$). Finally, we give a pricing scheme to incentivize each UE to adopt the designed $\bB$.

\section{Closed-form Average UE Data Rate}\label{section_performance}
In this section, we derive the  average UE data rate via stochastic geometric analysis.
 Consider a reference UE, termed as the \textit{typical  UE}, communicating with its  BS, termed as the \textit{typical BS}. %We first investigate the interference from all other mcrocell and femtocell UEs to the typical BS.
%We aim to investigate the performance of the typical UE.  Note that here we follow the conventional labeling of \textit{typicality} in the theory of stochastic geometry.
We are interested in the typical UE since the average UE performance in
the system is the same as the performance of the typical UE   \cite{SG_Totorial1}.
Furthermore, due to the stationarity of UEs and   BSs, throughout this section we will re-define the coordinates so that the typical  UE is located at  $\bzero$.

%\subsection{UE Association Probability} \label{subsection_probability}

%\subsection{Distance Distribution}

%\subsection{Laplace Transform of Interference and Coverage Probability}
First, we study the coverage probability given that the typical UE is associating with a tier-$k$ BS and their distance is $d$. In this case, the overall interference to the typical  UE is the sum interference from all tier-$k$ BSs other than the typical BS. %Because we assume that different tiers of BSs are working on separated spectrum, the interference from other tiers is zero.
Let $I_k(d)$ denote such interference. Then
 \begin{align}
 I_k(d)=\sum_{\bx \in \Phi_k'} \frac{P_kh_{\bx,\bzero}}{ \alpha|\bx|^{\gamma}}.
 \end{align}
 where $\Phi_k'$ is the Palm point process corresponding to all tier-$k$ BSs other than the typical BS, given that the typical BS is located at a distance of $d$ from the typical UE. It can be shown that $\Phi_k'$ is a PPP with intensity $0$ in $\mathcal{B}(\bzero, d)$ and intensity $\lambda_k$ in  $\mathbb{R}^2\backslash \mathcal{B}(\bzero, d)$, where $\mathcal{B}(\bzero, d)$ denotes the disk region centered at $\bzero$ with radius $d$ \cite{SG_Totorial1}.

 The distribution of  $I_k(d)$ is derived through its Laplace transform as follows:
\begin{align}
\nonumber&\mathcal{L}_{I_k}(d,s)=\mathbf{E}\left[\exp\left(-\sum_{\bx \in \Phi_k'} \frac{sP_kh_{\bx,\bzero}}{ \alpha|\bx|^{\gamma}}\right)\right]\\
\label{formula_Laplace1}=&\exp\left(-\lambda_k\int_{\mathbb{R}^2\backslash \mathcal{B}(\bzero, d)}\left(1-\int_{\mathbb{R}^+}e^{-\frac{sP_kh}{\alpha|\bx|^{\gamma}}}H(\dd h)\right)\dd\bx\right)\\
\label{formula_Laplace2}=&\exp\left(-\lambda_k\int_{\mathbb{R}^2\backslash \mathcal{B}(\bzero, d)}\frac{\frac{sP_k}{\alpha|\bx|^{\gamma}}}{\frac{sP_k}{\alpha|\bx|^{\gamma}}+1}\dd\bx\right)\\
%\label{formula_Laplace3}=&\exp\left(-2\pi\lambda_k\int_{d}^{\infty}\frac{\frac{sP_k}{\alpha r^{\gamma}}}{\frac{sP_k}{\alpha r^{\gamma}}+1}r\dd r\right)\\
\label{formula_LI}=&\exp\left(-2\pi\lambda_k\int_{d}^{\infty}\frac{\frac{sP_kr}{\alpha}}{\frac{sP_k}{\alpha}+r^{\gamma}}\dd r\right),
\end{align}
where (\ref{formula_Laplace1}) is obtained from the Laplace functional of PPP $\Phi_k'$ \cite{SG_Totorial1}, (\ref{formula_Laplace2}) is because $h$ is exponentially distributed with unit mean, and (\ref{formula_LI}) is through a transformation to polar coordinates.

 Let $P_{cover,k}(d)$ denote the conditional coverage probability of the typical UE (given $k$ and $d$). Then
\begin{align}
\nonumber P_{cover,k}(d) = &\mathbf{P}\left(\frac{P_kh_{\bx_B,\bzero}}{\alpha d^{\gamma}} \geq T I_k(d)\right) \\ \label{formula_Pc1}=&\mathcal{L}_{I_k}(d,s)|_{s=\frac{T\alpha d^{\gamma}}{P_k}},
\end{align}
where $\bx_B$ is the coordinate of the typical BS, and $|\bx_B|=d$.
Substituting  (\ref{formula_LI}) into (\ref{formula_Pc1}), we have
\begin{align}
\nonumber P_{cover,k}(d)
\label{formula_integral}=&\exp\left(-2\pi\lambda_k\int_{d}^{\infty}\frac{Td^{\gamma}r}{Td^{\gamma}+r^{\gamma}}\dd r\right)\\
 \overset {t=\frac{r^2}{T^{2/\gamma}d^2}}{=}&\exp\left(-\pi\lambda_kT^{\frac{2}{\gamma}}d^2\int_{(\frac{1}{T})^{\frac{2}{\gamma}}}^{\infty}\frac{1}{1+t^{\frac{\gamma}{2}}}\dd t\right).
\end{align}

Furthermore, the probability density function of the distance between the typical UE and its associated tier-$k$ BS is
\begin{align}
\label{formula_pdf1} f_k(d)=&\frac{2\pi \lambda_k}{A_k}d\exp\left(-\pi d^2\sum_{j=1}^{K}\lambda_j\left(\frac{P_jB_j}{P_kB_k}\right)^{\frac{2}{\gamma}}\right)\\
\label{formula_pdf2}=&\frac{2\pi \lambda_k}{A_k}d\exp\left(-\pi d^2\frac{\lambda_k}{A_k}\right),
\end{align}
where (\ref{formula_pdf1}) is derived in \cite{JF1}, and (\ref{formula_pdf2}) is by substituting (\ref{formula_Ak}) into (\ref{formula_pdf1}).

Hence, the coverage probability $P_{cover,k}$ of the typical UE associated with a tier-$k$ BS can be computed as
%
%\begin{align}
%\prod_{i=1}^{K}\exp\left(-2\pi\lambda_i\left(\frac{TP_i}{P_k}\right)^{\frac{2}{\gamma}}R^2\int_{\left(\frac{B_i}{TB_k}\right)^{\frac{2}{\gamma}}}^{\infty}\frac{1}{1+t^{\gamma/2}}dt\right)
%\end{align}
%Average coverage probability
\begin{align}
\nonumber P_{cover,k}=&\int_{0}^{\infty}f_k(d)P_{cover,k}(d) \dd d\\
\nonumber=&\int_{0}^{\infty}\frac{2\pi\lambda_k}{A_k} d\exp\left(-\pi d^2\frac{\lambda_k}{A_k}\right) \\ \nonumber&\qquad\exp\left(-\pi\lambda_k\left(T\right)^{\frac{2}{\gamma}}d^2\int_{\left(\frac{1}{T}\right)^{\frac{2}{\gamma}}}^{\infty}\frac{1}{1+t^{\gamma/2}}\dd t\right)\dd d\\
%=&\frac{\pi\lambda_k}{A_k}\frac{1}{\pi\frac{\lambda_k}{A_k}+\pi\lambda_k\left(T\right)^{\frac{2}{\gamma}}\int_{\left(\frac{1}{T}\right)^{\frac{2}{\gamma}}}^{\infty}\frac{1}{1+t^{\gamma/2}}\dd t}\\
=&\frac{1}{A_k}\frac{1}{\frac{1}{A_k}+C},
\end{align}
where $C=\left(T\right)^{\frac{2}{\gamma}}\int_{\left(\frac{1}{T}\right)^{\frac{2}{\gamma}}}^{\infty}\frac{1}{1+t^{\gamma/2}}\dd t$ is a system-level constant only related to $\gamma$ and $T$.

Let $\mathbf{E}_0(\beta_k)$ denote the expected spectrum bandwidth allocated to the typical UE (connecting to a tier-$k$ BS). Following the model in  Section \ref{subsection_allocation},
$\mathbf{E}_0(\beta_k)$ equals the spectrum bandwidth allocated to the typical tier-$k$ BS  divided by the average number of UEs associated with it conditioned on the typical UE, which is $A_k\mu/\lambda_k+1$. Hence,
\begin{align}
\mathbf{E}_0(\beta_k)=\frac{\eta_kW}{A_k\mu/\lambda_k+1}.
\end{align}
Then, by Section \ref{subsection_datarate}, the conditional expected data rate of the typical UE, given it is associated with a tier-$k$ BS, can be computed as \cite{JF1,Yicheng}\footnote{By doing so, we slightly underestimate the average data rate because the coverage event and $\beta_k$ are not completely independent.  Although some efforts have been made to approximate their correlation \cite{JF3, SizeDistribution}, all of them are inexact but result in tremendous mathematical complexity. %Thus,  we omit the correlation, similar to \cite{Yicheng}.
In Section \ref{section_experiment}, we  show that  the resultant analysis is close to actual performance via simulations. }
\begin{align}
\overline{R}_k=\mathbf{E}_0(\beta_k)\log(1+T)P_{cover,k}.
\end{align}

Finally, the average data rate of the typical UE, and hence the average data rate per UE in the system, is
\begin{align}
\nonumber \mathbf{F}=&\sum_{k=1}^{K}A_k\overline{R}_k=\sum_{k=1}^{K}A_k\mathbf{E}_0(\beta_k)\log(1+T)P_{cover,k}\\
=&\sum_{k=1}^{K}\frac{\eta_kW\log(1+T)}{(A_k\mu/\lambda_k+1)(\frac{1}{A_k}+C)}.
\end{align}

Note that stochastic geometric analysis often leads to non-closed forms requiring numerical integrations (e.g., \cite{JF4, JF5, JF1,JF3}), due to the integral form of the Laplace functional or generating functional of PPPs applied in analysis \cite{SG_Totorial1,SG_book1}. Fortunately, our derived closed-form expression for the average UE data rate facilitates the tractability of the resultant optimization problem.

\section{Joint Optimization Problem and SSAUA}\label{section_optimization}

We aim to maximize the average UE data rate $\mathbf{F}$ with respect to $\bfeta$ and $\bB$.  As there is a one-to-one mapping between $\bA$ and $\bB$, we study the optimization problem over $(\bfeta, \bA)$ instead  for analytical convenience.
This is formally stated as  optimization problem $\mathbf{P}$  as follows:
\begin{align}
\nonumber& \underset{\bfeta,\bA}{\text{maximize}}
& & \mathbf{F}(\bfeta,\bA)=\sum_{k=1}^{K}\eta_kM_k(A_k) \\
\nonumber& \text{subject to}
\nonumber& & \sum_{k=1}^{K}\eta_k=1,\quad \eta_{\min,k}\leq\eta_k\leq\eta_{\max,k},\forall k,\\
\label{formula_optimal}& & & \sum_{k=1}^{K}A_k=1,\quad A_k\geq0,\forall k,
\end{align}
where $M_k(A_k)$ is defined as
\begin{align}
M_k(A_k)=\frac{1}{\left(A_k\mu/\lambda_k+1\right)\left(\frac{1}{A_k}+C\right)}.
\end{align}

Problem $\mathbf{P}$ is  non-convex and cannot be solved through a standard method.
Instead, we investigate into two important structures of the optimal solution, termed density thresholding and priority ordering, based on which we propose a computationally efficient Structured Spectrum Allocation and User Association (SSAUA) approach to solve  the problem.

\subsection{Density Thresholding Structure}\label{subsection_two}

First, we define an important parameter
\begin{align}\label{formula_a}
a_k\triangleq\sqrt{\lambda_k/(\mu C)}.
\end{align}
Note that $M_k(A_k)$ is increasing on $[0,a_k]$ and decreasing on $[a_k,\infty)$.
We further observe several useful properties of $M_k(A_k)$, which are presented in Appendix \ref{appendix-Mk}. Based on these properties, we obtain the following lemma,  whose proof is given in Appendix \ref{appendix-lemma1}.

\begin{lemma}\label{lemma1}
Consider  a potential solution $(\betass, \bAss)$ to Problem $\mathbf{P}$. If $\exists i\neq j$, such that $\Ass_i<a_i$ and $\Ass_j>a_j$, then $(\betass, \bAss)$ is not an optimal solution.
\end{lemma}
%\begin{proof}
%See Appendix-\ref{appendix-lemma1} for the proof.
%\end{proof}

Lemma \ref{lemma1} suggests that, in an optimal solution, every
$A_k$ must be on the same side of $a_k$.  This directly leads  to the following theorem, which is fundamental to our optimization solution.
\begin{theorem}\label{theorem_branch}
\textbf{(Density Thresholding)} Let $(\betas, \bAs)$ be an optimal solution to Problem $\mathbf{P}$. If $\sum_{k=1}^{K}a_k>1$, then $\forall k, \As_k\leq a_k$; if $\sum_{k=1}^{K}a_k<1$, then $\forall k, \As_k\geq a_k$;   if $\sum_{k=1}^{K}a_k=1$, then $\forall k, \As_k= a_k$.
\end{theorem}
\begin{proof}
If $\sum_{k=1}^{K}a_k>1$,  because $\sum_{k=1}^{K}\As_k=1$, $\exists l$ such that $\As_l<a_l$. This leads to  $\As_k\leq a_k$, $\forall k$, according to Lemma \ref{lemma1}.
The cases where  $\sum_{k=1}^{K}a_k<1$ and $\sum_{k=1}^{K}a_k=1$ are similar.
%
%If $\sum_{k=1}^{K}\sqrt{\frac{1}{a_k}}>1$,  $\exists k$ such that $\As_k<\sqrt{\frac{1}{a_k}}$. This leads to  $\As_k\leq\sqrt{\frac{1}{a_k}}$, $\forall k$ according to Lemma \ref{lemma1}.
%
%If $\sum_{k=1}^{K}\sqrt{\frac{1}{a_k}}=1$,  $\forall k, \As_k= \sqrt{\frac{1}{a_k}}$. Otherwise $\exists i\neq j$ such that $\As_i<\sqrt{\frac{1}{a_i}}$ and $\As_j<\sqrt{\frac{1}{a_j}}$, contradicting Lemma \ref{lemma1}.
\end{proof}

Note that, the condition $\sum_{i=1}^{K}a_i>1$ is equivalent to $\sqrt{\frac{1}{ C} }\left(\sum_{i=1}^{K}\sqrt{\lambda_i}\right)>\sqrt{\mu}$,
 implying the density of UEs is sparse (compared with that of BSs). Thus, we refer to the case $\sum_{i=1}^{K}a_i>1$ as the \emph{sparse UE scenario}. On the other hand, we refer to the case  $\sum_{i=1}^{K}a_i<1$, which is equivalent to $\sqrt{\frac{1}{ C}}\left(\sum_{i=1}^{K}\sqrt{\lambda_i}\right)<\sqrt{\mu}$,  as the \emph{dense UE scenario}. If $\sum_{i=1}^{K}a_i=1$, Problem $\mathbf{P}$ can  be trivially solved and is ignored in the rest of our discussion. Note that because $a_k$ can be computed directly from the given parameters, one can judge which scenario Problem $\mathbf{P}$ falls within before solving the problem.
Next, the solution to $\mathbf{P}$ will be investigated separately in the sparse UE and dense UE scenarios.

\subsection{SSAUA  in the Sparse UE Scenario}\label{subsection_sparse}
 In this case, the original Problem $\mathbf{P}$ becomes Problem $\mathbf{P1}$ as follows:
\begin{align}
\nonumber& \underset{\bfeta, \bA}{\text{maximize}}
%& & \mathbf{F}(\bfeta, \mathbf{A})=\sum_{k=1}^{K}\frac{\eta_k}{\left(A_k\mu/\lambda_k+1\right)
%\left(\frac{1}{A_k}+C\right)} \\
& & \mathbf{F}(\bfeta, \mathbf{A})=\sum_{k=1}^{K}\eta_k M_k(A_k)\\
\nonumber& \text{subject to}
\nonumber& & \sum_{k=1}^{K}\eta_k=1, \quad\eta_{\min,k}\leq\eta_k\leq\eta_{\max,k},\forall k,\\
\label{formula_optimalP1}& & & \sum_{k=1}^{K}A_k=1,\quad0\leq A_k\leq a_k,\forall k.
\end{align}

We first observe an important ordering property of the optimal solution to $\mathbf{P1}$, as shown in the following lemma, whose proof is given in Appendix \ref{appendix-lemma_order1}.
\begin{lemma}\textbf{(Ordering Property)}\label{lemma_order1}
Let  $\bAs$ be  optimal for $\mathbf{P1}$, then $M_1(\As_1)\leq M_2(\As_2)\leq \ldots\leq M_K(\As_K)$.
\end{lemma}

%\begin{proof}
%See Appendix-\ref{appendix-lemma_order1} for the proof.
%\end{proof}

%Lemma \ref{lemma_order1} shows an important ordering structure of optimal solutions of $\mathbf{P1}$. We aim to find an optimal solution of $\mathbf{P1}$ by using the ordering property.

Next, by sequentially computing $\betas$ as follows:
\begin{align}\label{formula_solutioneta}
\left\{
  \begin{array}{l}
    \etas_K=\min(1-\sum_{k=1}^{K-1}\eta_{\min,k},\eta_{\max,K}), \\
    \etas_{K-1}=\min(1-\etas_K-\sum_{k=1}^{K-2}\eta_{\min,k},\eta_{\max,K-1}), \\
    \ldots,\\
     \etas_{l}=\min(1-\sum_{k=l+1}^{K}\etas_k-\sum_{k=1}^{l-1}\eta_{\min,k},\eta_{\max,l}), \\
    \ldots,\\
     \etas_{1}=\min(1-\sum_{k=2}^{K}\etas_k,\eta_{\max,1}), \\
  \end{array}
\right.
\end{align}
we have the following theorem:
\begin{theorem}\label{theorem1}
\textbf{(Priority Ordering)} Let $\bAs$ be  optimal for Problem $\mathbf{P1}$, then $(\betas, \bAs)$, where $\betas$ is computed in (\ref{formula_solutioneta}),  is an optimal solution to $\mathbf{P1}$.
\end{theorem}
\begin{proof}
 Consider Problem $\mathbf{P1A}$ as follows:
\begin{align}
\nonumber& \underset{\bfeta}{\text{maximize}}
& & \sum_{k=1}^{K}\eta_k M_k(\As_k) \\
& \text{subject to}
& & \sum_{k=1}^{K}\eta_k=1,\quad\eta_{\min,k}\leq\eta_k\leq\eta_{\max,k},\forall k.
\end{align}

Then $\mathbf{P1A}$ is a simple linear programming problem with ordered linear coefficients in the objective, since $M_1(\As_1)\leq M_2(\As_2)\leq \ldots\leq M_K(\As_K)$ due to Lemma \ref{lemma_order1}. Note that $\betas$ does not depend on the exact values of $\bAs$; it only requires the ordering property as shown in Lemma \ref{lemma_order1}. Also, $\betas$ is in the feasible region due to $\sum_{k=1}^{K} \eta_{\min,k} \leq 1\leq \sum_{k=1}^{K} \eta_{\max,k}$.
It is easy to verify that, $(\betas, \bAs)$ is an optimal solution to $\mathbf{P1}$.
\end{proof}

%Note that because $M_1(\As_1)\leq M_2(\As_2)\leq \ldots\leq M_K(\As_K)$, the optimal $\betas$ derived in (\ref{formula_solutioneta}) shows the \textit{priority ordering structure in the spectrum allocation}. In this case, the spectrum is firstly allocated to tier $K$ BS as large as possible, then allocated to tier $K-1$, and so forth, until all the spectrum is allocated. In this sense, a tier with higher BS density has a higher priority in allocating the spectrum.

Equation (\ref{formula_solutioneta}) indicates the priority ordering structure in spectrum allocation. We see that tier-$K$ has the highest priority in spectrum allocation, followed by tier-$(K-1)$, and so forth.

Theorem \ref{theorem1} provides a means to derive an optimal $\betas$ regardless of the $\bAs$ values. We need one  further step to derive the corresponding optimal $\bAs$ by solving the following Problem $\mathbf{P1B}$:
\begin{align}
\nonumber& \underset{\bA}{\text{maximize}}
& & \sum_{k=1}^{K}\etas_k M_k(A_k) \\
& \text{subject to}
& & \sum_{k=1}^{K}A_k=1, \quad 0\leq A_k\leq a_k,\forall k.
\end{align}
Note that $\mathbf{P1B}$ is a  convex programming problem, since $M_k(A_k)$ is concave on $[0,a_k]$. Thus, $\bAs$ can be computed by a computationally efficient algorithm, such as the interior point method. Hence the both steps to compute the jointly optimal solution $(\betas, \bAs)$ have low computational complexity.

%In summary, in the sparse UE scenario, through our SSAUA approach,

\subsection{SSAUA in the Dense UE Scenario with Performance Bound}\label{subsection_dense}
%In this subsection, we discuss the dense UE scenario (i.e., $\sqrt{\frac{1}{C}}\left(\sum_{i=1}^{K}\sqrt{\lambda_i}\right)<\sqrt{\mu}$).
In this case, the original Problem $\mathbf{P}$ becomes Problem $\mathbf{P2}$ as follows:
\begin{align}
\nonumber& \underset{\bfeta,\bA}{\text{maximize}}
%& & \mathbf{F}(\bfeta, \mathbf{A})=\sum_{k=1}^{K}\frac{\eta_k}{\left(A_k\mu/\lambda_k+1\right)
%\left(\frac{1}{A_k}+C\right)} \\
& & \mathbf{F}(\bfeta, \mathbf{A})=\sum_{k=1}^{K}\eta_kM_k(A_k)\\
\nonumber& \text{subject to}
\nonumber& & \sum_{k=1}^{K}\eta_k=1,\quad\eta_{\min}\leq\eta_k\leq\eta_{\max},\forall k,\\
\label{formula_optimalP2}& & & \sum_{k=1}^{K}A_k=1, \quad A_k\geq a_k,\forall k.
\end{align}
Problem $\mathbf{P2}$ is more complicated compared with Problem $\mathbf{P1}$, as $M_k(A_k)$ is not concave,
 but an S-shaped function, in the feasible region. Hence, $\mathbf{P2}$ generally incurs high computational complexity even if an optimal $\betas$ is given
\cite{Optimization_SB, Optimization_MC}.

Therefore, instead of directly solving $\mathbf{P2}$, we first approximate $M_k(A_k)$ by $\Mb_k(A_k)$ defined as follows:
\begin{align}
\Mb_k(A_k)=\frac{1}{\left(A_k\mu/\lambda_k\right)\left(\frac{1}{A_k}+C\right)}.
\end{align}
Note that this approximation is reasonable because $A_k\mu/\lambda_k$ is much larger than $1$ when $\mu$ is large (i.e., the dense UE scenario). This observation is also supported by the small performance gap as derived in Section \ref{subsubsection_bound}.  Some useful properties of $\Mb_k(A_k)$ are shown in Appendix \ref{appendix-Mbk}.

The approximated problem is referred to as Problem $\mathbf{P2A}$, where we simply replace the objective function of $\mathbf{P2}$ by the following:
%\begin{align}
%\nonumber& \underset{\bfeta,\bA}{\text{maximize}}
%& & \mathbf{F'}(\bfeta,\bA)=\sum_{k=1}^{K}\eta_k\Mb_k(A_k) \\
%\nonumber& \text{subject to}
%\nonumber& & \sum_{k=1}^{K}\eta_k=1,\quad\eta_{\min,k}\leq\eta_k\leq\eta_{\max,k},\forall k,\\
%\label{formula_optimalp2a}& & & \sum_{k=1}^{K}A_k=1,\quad A_k\geq\sqrt{\frac{1}{a_k}},\forall k.
%\end{align}
\begin{align}
\label{formula_optimalp2a} \mathbf{F'}(\bfeta,\bA)=\sum_{k=1}^{K}\eta_k\Mb_k(A_k).
\end{align}

\subsubsection{Solution to $\mathbf{P2A}$}

The important  ordering property still holds for Problem $\mathbf{P2A}$, as formalized in the following lemma, whose proof is given in Appendix \ref{appendix-lemma_order2}.

\begin{lemma}\label{lemma_order2}
\textbf{(Ordering Property)} Let $(\betasb, \bAsb)$ be an optimal solution to $\mathbf{P2A}$, then $\Mb_1(\Asb_1)\leq \Mb_2(\Asb_2)\leq \ldots\leq \Mb_K(\Asb_K)$.
\end{lemma}

We observe that with the same ordering property,  (\ref{formula_solutioneta}) can again be adopted as an optimal solution to $\mathbf{P2A}$ in the dense UE scenario, leading to the following theorem:
\begin{theorem}\label{theorem2}
\textbf{(Priority Ordering)} Let $\bAsb$ be  optimal for Problem $\mathbf{P2A}$, then $(\betasb, \bAsb)$, where $\betasb$ is computed the same way as $\betas$ in (\ref{formula_solutioneta}),  is an optimal solution to $\mathbf{P2A}$.
\end{theorem}
\begin{proof}
The proof is similar to that of Theorem \ref{theorem1}.
\end{proof}

Given an optimal $\betasb$ for $\mathbf{P2A}$, we find the corresponding optimal $\bAsb$ for $\mathbf{P2A}$ by solving the following Problem $\mathbf{P2B}$:
\begin{align}
\nonumber& \underset{\bA}{\text{maximize}}
& & \sum_{k=1}^{K}\etasb_k \Mb_k(A_k) \\
& \text{subject to}
         & & \sum_{k=1}^{K}A_k=1,\quad A_k\geq a_k,\forall k.
\end{align}
Unlike in the sparse UE scenario, here we have an explicit solution, as stated in the following theorem:
\begin{theorem}\label{theorem_order3}
Given an optimal $\betasb$ for $\mathbf{P2A}$ (computed the same way as $\betas$ in (\ref{formula_solutioneta})), the corresponding optimal $\bAsb$ can be expressed as follows:
\begin{align}\label{formula_solutionA}
\left\{
  \begin{array}{ll}
    \Asb_k=a_k, & k\geq 2 \\
     \Asb_{1}=1-\sum_{l=2}^{K}\Asb_l. \\
  \end{array}
\right.
\end{align}
\end{theorem}
\begin{proof}
See Appendix \ref{appendix-theorem_order3}.
\end{proof}
Note that both  (\ref{formula_solutioneta}) and (\ref{formula_solutionA}) can be computed with low computational complexity.

\subsubsection{Bounding the Performance Gap}\label{subsubsection_bound}
Since $(\betasb, \bAsb)$ is optimal for $\mathbf{P2A}$ rather than $\mathbf{P2}$,  we next quantify the performance gap between $(\betasb, \bAsb)$ and an optimal solution $(\betas, \bAs)$ to $\mathbf{P2}$.

The performance gap is defined as
\begin{align}
E=\mathbf{F}(\betas,\bAs)-\mathbf{F}(\betasb,\bAsb).
\end{align}
Because $\mathbf{F}(\betasb,\bAsb)\leq \mathbf{F}(\betas,\bAs)\leq  \mathbf{F}'(\betas,\bAs)\leq\mathbf{F}'(\betasb,\bAsb)$, we have
\begin{align}
E\leq \mathbf{F}'(\betasb,\bAsb)-\mathbf{F}(\betasb,\bAsb)\triangleq E'.
\end{align}
Substituting $\betasb$ and $\bAsb$ into $E'$, we have
\begin{align}
E'=&\sum_{k=1}^{K}\frac{\etasb_k}{\frac{\Asb_k\mu}{\lambda_k}\left(\frac{\Asb_k\mu}{\lambda_k}+1\right)\left(\frac{1}{\Asb_k}+C\right)}.%\\
%\nonumber \leq&\sum_{k=1}^{K}\frac{\etasc_k}{\sqrt{\frac{\mu}{\lambda_k C}}\left(\sqrt{\frac{\mu}{\lambda_k C}}+1\right)\left(\sqrt{\frac{\mu C}{\lambda_k}}+C\right)}
\end{align}
Therefore, the relative performance gap is bounded by
\begin{align}
\epsilon\triangleq&\frac{E}{\mathbf{F}(\betas,\bAs)}\leq \frac{E'}{\mathbf{F}(\betasb,\bAsb)}\\
\label{formula_bound}=&\frac{\sum_{k=1}^{K}\frac{\etasb_k}{\frac{\Asb_k\mu}{\lambda_k}\left(\frac{\Asb_k\mu}{\lambda_k}+1\right)\left(\frac{1}{\Asb_k}+C\right)}}
{\sum_{k=1}^{K}\frac{\etasb_k}{\left(\frac{\Asb_k\mu}{\lambda_k}+1\right)\left(\frac{1}{\Asb_k}+C\right)}}\\
\stackrel{(a)}{\leq} &\max_k\frac{\lambda_k}{\Asb_k\mu}\leq \max_k \frac{\lambda_k}{a_k\mu}
\\ = &\sqrt{\lambda_KC/\mu},
\end{align}
where inequality (a) is obtained by observing the common factor in the summations in the numerator and denominator of (\ref{formula_bound}). This implies that $\epsilon$ scales as $O\left(\sqrt{\lambda_K/\mu}\right)$.
Note that because we are considering the dense UE scenario (i.e., $\sqrt{\frac{1}{C}}\left(\sum_{i=1}^{K}\sqrt{\lambda_i}\right)<\sqrt{\mu}$),  $O\left(\sqrt{\lambda_K/\mu}\right)$ is small by definition.

\subsection{Estimation for Complexity of Exhaustive Search}
In this subsection, we briefly discuss the complexity of the exhaustive search approach to solve Problem $\mathbf{P}$. First, as explained in details in our technical report \cite{OurReport},  we observe that at least  one of the optimal solutions to  $\mathbf{P}$, $(\betas, \bAs)$,  has the following property:  there is at most one $k\in\{1,2,\ldots K\}$ such that $\eta_{\min,k}<\etas_k<\eta_{\max,k}$; $\forall j\neq k$, either $\etas_j=\eta_{\min,j}$ or $\etas_j=\eta_{\max,j}$ (i.e., at the boundary). Thus, the search for $\betas$ needs to be performed only at these boundary cases, leading to a complexity of $\Omega(2^K)$.  Furthermore, in the dense UE scenario, a numerical search  over all locally optimal $\bA$ (at least $2^K$ of them) is required, leading to another fold of  $\Omega(2^K)$, i.e., $\Omega(4^K)$ overall, in complexity.
Numerical studies on the computational complexity will be presented in Section \ref{section_experiment}.

\section{Nash Equilibrium for  SSAUA}\label{section_equilibirum}
Individual UEs may behave selfishly to derive unfair advantage despite our design of $\bBs$ (or equivalently $\bAs$). Thus, in this section, we propose a Surcharge Pricing Scheme (SPS), such that the designed $\bBs$ is the natural outcome of a Nash equilibrium.  Note that the designed spectrum allocation factors $\betas$ can be maintained by the network operator and is beyond our concern.

We consider a reference individual UE, whose association bias values are $\bB'=(B_1', B_2',\ldots, B_K')$.  Let $\bA'=(A_1', A_2',\ldots, A_K')$ be its corresponding association probabilities. For the  other UEs, suppose they all obey the association bias values $\bBs$ assigned by the network operator.
Similar to the discussions in Section \ref{section_model} and \ref{section_performance}, %$A_k'=\frac{\lambda_k(P_kB_k')^{\frac{2}{\gamma}}}{\sum_{j=1}^{K}\lambda_j(P_jB_j')^{\frac{2}{\gamma}}}$.
%Let $P_{cover,k}'$ denote the coverage probability of the  reference UE given it associates itself with a tier-$k$ BS, then  we have $P_{cover,k}'=\frac{1}{A_k'}\frac{1}{\frac{1}{A_k'}+C}$.
%Let $\mathbf{E}_0(\beta_k')$ be the expected spectrum bandwidth allocated to the reference UE.  Then $\mathbf{E}_0(\beta_k')=\frac{\etas_kW}{\As_k\mu/\lambda_k+1}.$
 %Therefore,
 the average data rate of the reference UE  is
%\begin{align}
%\nonumber \mathbb{F}=&\sum_{k=1}^{K}A_k'\mathbf{E}_0(\beta_k')\log(1+T)P_{cover,k}'\\
%=&\sum_{k=1}^{K}\frac{\etas_kW\log(1+T)}{(\As_k\mu/\lambda_k+1)(\frac{1}{A_k'}+C)}.
%\end{align}
\begin{align}
\mathbb{F}=\sum_{k=1}^{K}\frac{\etas_kW\log(1+T)}{(\As_k\mu/\lambda_k+1)(\frac{1}{A_k'}+C)}.
\end{align}

If the reference UE performs an optimization on $\mathbb{F}$ with respect to $\bA'$,  the resultant optimal $\bAsd=(\Asd_1,\Asd_2,\ldots, \Asd_K)$ is  unlikely to be the same as $\bAs$.
Therefore, we add the following Surcharge Pricing Scheme: the network operator applies a surcharge $c_k$ to each UE associated with a tier-$k$ BS. Let $\bc=(c_1,c_2,\ldots, c_K)$. In this case, the average surcharge for the reference UE is $\sum_{k=1}^{K}c_kA_k'$.
Accordingly, the  reference UE  will perform the following optimization Problem $\mathbf{P3}$:
\begin{align}
\nonumber& \underset{\bA'}{\text{maximize}}
& & \mathbb{F}'=\sum_{k=1}^{K}\left(\frac{ \etas_kW\log(1+T)}{\left(\As_k\mu/\lambda_k+1\right)
\left(\frac{1}{A_k'}+C\right)}-c_k A_k'\right) \\
& \text{subject to}
& & \sum_{k=1}^{K}A_k'=1, \quad A_k'\geq 0.
\end{align}
%where $\rho$ is a given positive tradeoff factor \cite{game_FB}, i.e., the importance each UE assigns to its data rate relative to the surcharge.

Different from $\mathbf{P}$, it can be shown that $\mathbf{P3}$ is a standard convex optimization problem. By the KKT conditions, its optimal solution $\bAsd$ satisfies %
%\begin{align}
%\frac{\partial R}{\partial A_k'}-\nu=0
%\end{align}
\begin{align}
\label{formula_KKT1}\frac{H_k}{(1+C\Asd_k)^2}-c_k-\nu+\theta_k=0,\\
\theta_k\Asd_k=0, \quad\theta_k\geq0,
\end{align}
where $H_k=\frac{ \etas_k W\log(1+T)}{\As_k\mu/\lambda_k+1}$, $\theta_k$ is a Lagrange multiplier corresponding to the inequality constraint $A_k'\geq 0$, and $\nu$ is a Lagrange multiplier corresponding to the equality constraint $\sum_{k=1}^{K}A_k'=1$.

Setting $\Asd_k=\As_k$, we have
\begin{align}
c_k=\begin{cases}
\infty, & \textrm{if } \As_k=0,\\
\frac{H_k}{(1+C\As_k)^2}-\nu, & \textrm{otherwise}.
\end{cases}
\end{align}
Note that $\nu$ could be set arbitrarily due to the equality constraint. Without loss of generality, we set $\nu=\min_k\frac{H_k}{(1+C\As_k)^2}$ so that the minimum surcharge among tiers is $0$.
As a consequence, a Nash Equilibrium is achieved where every UE adopts the assigned $\bBs$.

\section {Numerical Study} \label{section_experiment}

\begin{figure}[t]
\centering  \hspace{0pt}
\includegraphics[scale=0.41]{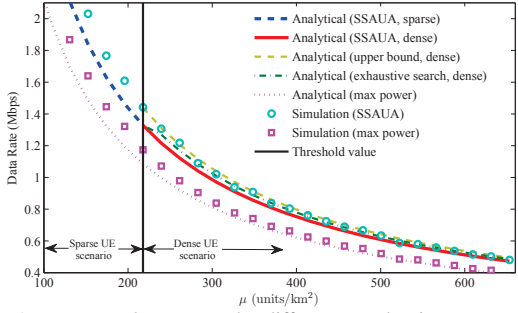}
\vspace{-0.4 cm}
\caption{Average UE data rate under different UE density $\mu$.}
\label{fig1}
\end{figure}

\begin{figure} \centering
\addtolength{\subfigcapskip}{-0.3cm}
\subfigure[Sparse UE scenario.]{
\includegraphics[scale=0.41]{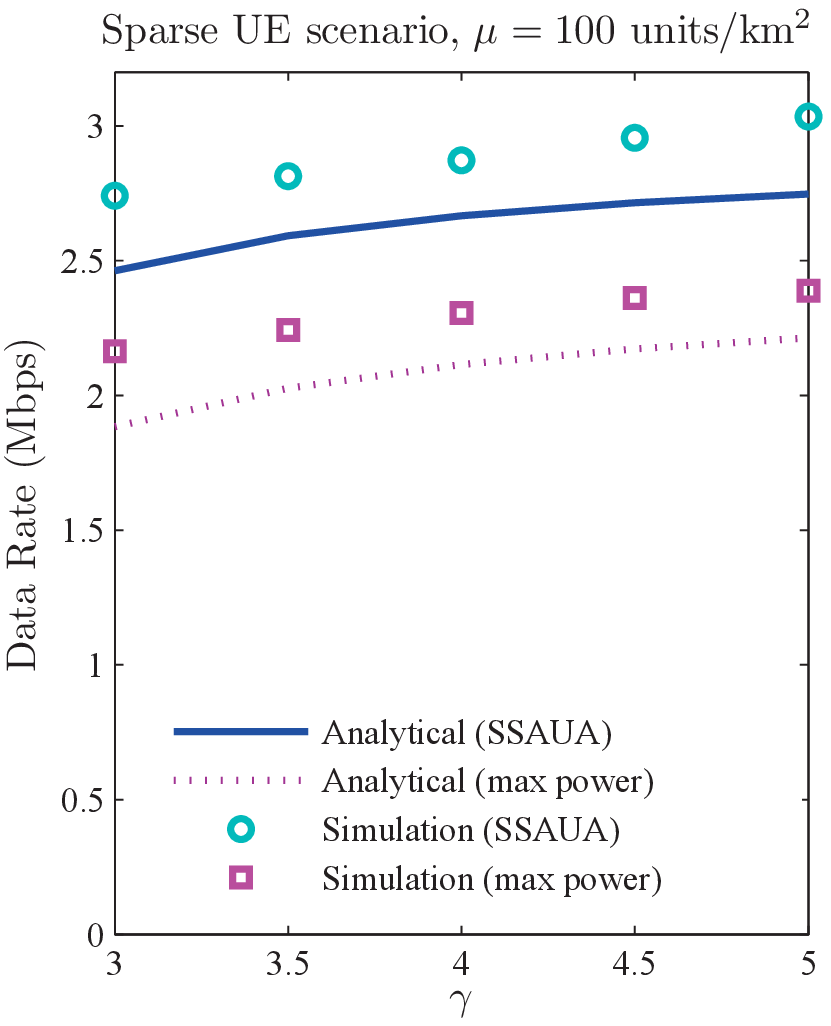}
}
\subfigure[Dense UE scenario.] {
\includegraphics[scale=0.41]{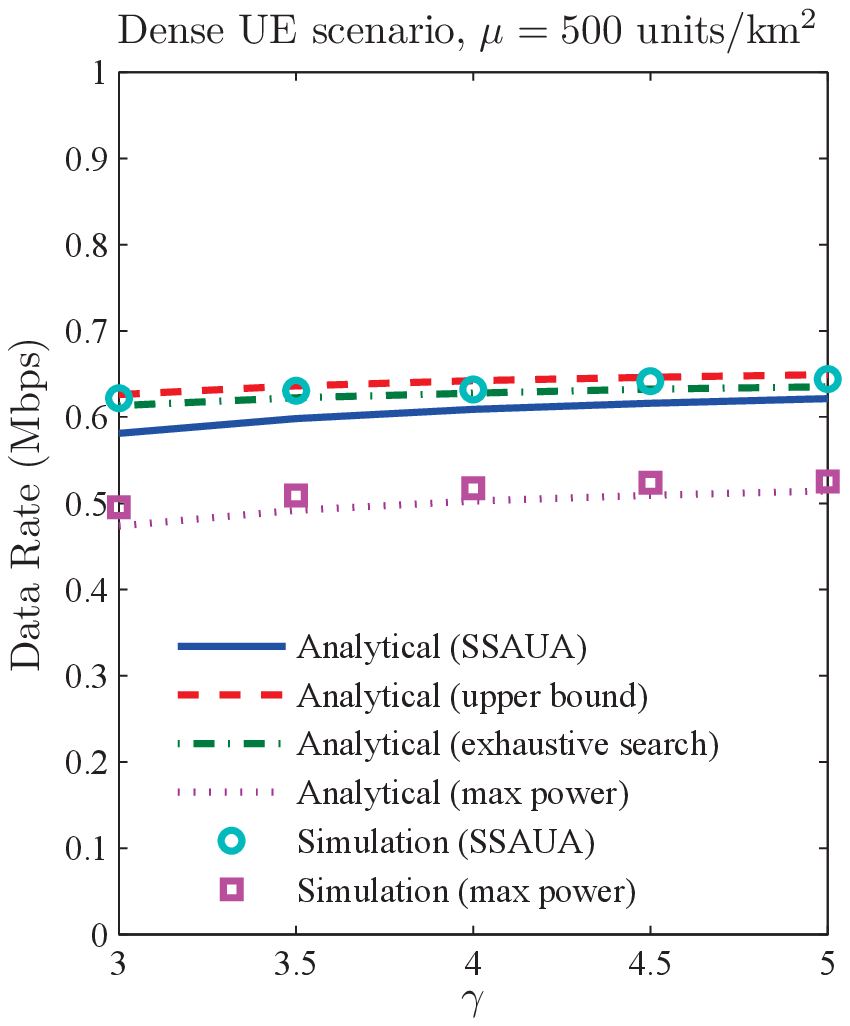}
}
\vspace{-0.3 cm}
\caption{Average UE data rate under different path loss exponent $\gamma$.}
\label{fig1_2}
\end{figure}

\begin{figure}[t]
\centering  \hspace{0pt}
\includegraphics[scale=0.41]{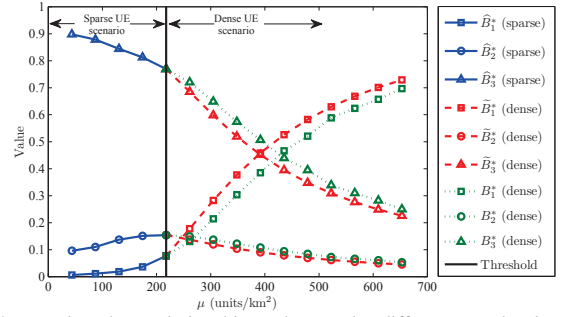}
\vspace{-0.4 cm}
\caption{Designed association bias values under different UE density $\mu$.}
\label{fig2}
\end{figure}

\begin{figure}[t]
\centering  \hspace{0pt}
\includegraphics[scale=0.41]{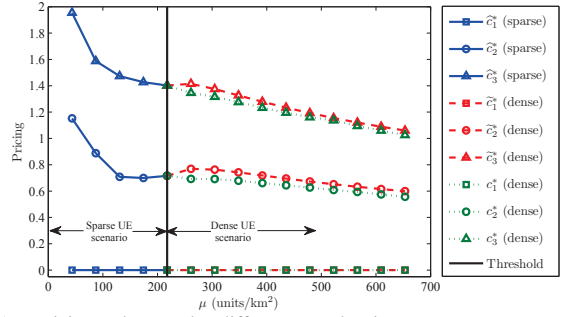}
\vspace{-0.4 cm}
\caption{Pricing values under different UE density $\mu$.}
\label{fig3}
\end{figure}

\begin{figure}[t]
\centering  \hspace{0pt}
\includegraphics[scale=0.41]{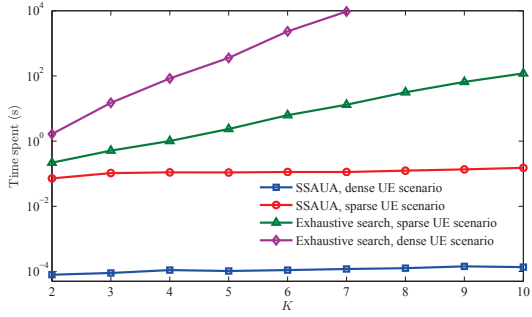}
\vspace{-0.4 cm}
\caption{Comparison of run time.}
\label{fig4}
\end{figure}

In this section,  we present numerical studies on the performance of SSAUA.
 We label the SSAUA solution as $(\betasc, \bBsc)$ and $(\betasb, \bBsb)$ in the sparse and dense UE scenarios, respectively.  Note that $(\betasc, \bBsc)$ is optimal in the sparse UE scenario. We also compare $(\betasb, \bBsb)$ with an optimal solution $(\betas, \bBs)$ obtained from exhaustive search in the dense UE scenario.

%$\bcs=(c_1^*,c_2^*,\ldots,c_K^*)$ is the pricing values of the  SPS computed  based on $(\betas, \bBs)$ and  $\bcsc=(\widehat{c}_1^*,\widehat{c}_2^*,\ldots,\widehat{c}_K^*)$ is computed  based on $(\betasc, \bBsc)$. Some notations have slightly different meanings compared with those used in Section \ref{section_optimization} and \ref{section_equilibirum}.

First, we study the network performance under different UE density $\mu$.
The network parameters are as follows: $K=3$, $\lambda_1=1$ units/km$^2$, $\lambda_2=5$ units/km$^2$, $\lambda_3=10$ units/km$^2$, $P_1=56$ dBm, $P_2=46$ dBm, $P_3=36$ dBm, $\eta_{\min,1}=0.2$, $\eta_{\min,2}=0.25$, $\eta_{\min,3}=0.3$,  $\eta_{\max,1}=0.35$, $\eta_{\max,2}=0.4$, $\eta_{\max,3}=0.45$, $\gamma=4$, $W=200$ MHz, and $T=0.2$. In each round of simulation,  UEs and BSs are generated on a $10$ km $\times$ $10$ km square, and the UEs in the central  $5$ km $\times$ $5$ km square are sampled  for performance evaluation (in order to remove the edge effect). Each simulation data point is averaged over all sampled UEs during 100 rounds of simulations.

 The results are shown in Fig.~\ref{fig1}. A vertical line indicates the threshold value of $\mu$, as given in Theorem \ref{theorem_branch}, separating the sparse and dense UE scenarios. For both scenarios, we show results of the analytical and simulated performance of SSAUA, and analytical and simulated performance of a ``max power'' approach, which employs equal spectrum allocation and user association based on the maximum received power. Since SSAUA  is not optimal in the dense UE scenario, we also add two sets of results accordingly:
the analytical optimal performance $\mathbf{F}(\betas,\bBs)$ through exhaustive search and its analytical upper bound  $\mathbf{F'}(\betasb,\bBsb)$.

Fig.~\ref{fig1} illustrates that the SSAUA achieves near-optimal solution in the dense UE scenario. Furthermore, SSAUA substantially outperforms the max power solution. Finally, the analytical performance is only slightly smaller than the simulated performance,  matching our discussions in Section \ref{section_performance}.

Fig.~\ref{fig1_2} shows the optimal network performance under different path loss exponent $\gamma$. The network parameters are the same as those used in Fig.~\ref{fig1} except $\mu$ is fixed at $100$ (i.e., sparse UE) and $500$ (i.e., dense UE) units/km$^2$ in Fig.~\ref{fig1_2} (a) and (b) respectively. This figure further confirms the observations from Fig.~\ref{fig1}.  Furthermore, it shows that SSAUA is effective for a wide range of path loss conditions.

Fig.~\ref{fig2} shows $\bBsc$, $\bBsb$, and $\bBs$; and Fig.~\ref{fig3} shows their corresponding prices $\bcsc$, $\bcsb$, and $\bcs$,  under different $\mu$. %In the dense UE scenario,
  We observe that $\bBsb$ and $\bcsb$ computed based on SSAUA approach are  close to their counterparts  $\bBs$ and $\bcs$.

 Finally, a run time experiment is conducted to compare   the computational complexity of SSAUA with that of exhaustive search. The experiment is executed by Matlab R2011a on an ASUS PC with Intel i7-3610QM 2.3GHz processor and 4GB RAM.
 The results are averaged  over 1000 runs for SSAUA and 10 runs for exhaustive search (both with randomly generated parameters).
 Fig.~\ref{fig4} shows that the run time of SSAUA is almost negligible compared with exhaustive search. Note that the $y$-axis is in log scale. When $K$ increases, the run time of exhaustive search exhibits an exponential growth tendency, while SSAUA remains computationally efficient.

\section{Conclusions}\label{section_conclusion}
In this work, we provide a theoretical  framework to study the joint optimization of spectrum allocation and user association  in heterogeneous cellular networks. We establish a stochastic geometric model that captures the random spatial patterns of BSs and UEs, and a closed-form expression of the analytical average UE data rate is derived.
We then consider the problem of maximizing the average UE data rate by jointly optimizing the spectrum allocation factors $\bfeta$ and user association bias values $\bB$, which is non-convex programming in nature. We propose the SSAUA approach to solve this problem with low computational complexity. We show that the SSAUA approach is optimal in the sparse UE scenario and near-optimal in the dense UE scenario, with a quantified tight bound scaling as $O(\sqrt{\lambda_K/\mu})$. We also propose the SPS such that the designed association bias values can be achieved in  Nash equilibrium.

\appendix
\subsection{Useful Properties of $M_k(A_k)$ }\label{appendix-Mk}
%Here we present some important properties of $M_k(A_k)$.
\begin{description}
  \item [(M-1) ] $M_k(A_k)$ is increasing on $[0,a_k]$ and decreasing on $[a_k,\infty)$.
  \item [(M-2) ]  $M_k(A_k)$ is concave on $[0,a'_k]$ and convex on $[a'_k,\infty)$, where $a'_k$ is some threshold value $a'_k>a_k$. $M_k(A_k)$ is concave on $[0,a_k]$.
  \item [(M-3) ] If $\lambda_i<\lambda_j$, then $M_i(A)<M_j(A), \forall A>0$.
  \item [(M-4) ] If $\lambda_i<\lambda_j$, then $M_j(A)-M_i(A)$ is strictly increasing on $[0,a_j]$. %(i.e., $M_j(A)-M_i(A)>M_j(A')-M_i(A')$ if $A>A'$).
\end{description}

\subsection{Proof of Lemma \ref{lemma1}}\label{appendix-lemma1}
\begin{proof}
Suppose $\bAss$ is optimal, $\Ass_i<a_i$, and $\Ass_j>a_j$. Consider that we increase $\Ass_i$ by a small value $\Delta>0$ and decrease $\Ass_j$ by $\Delta$. According to property (M-1), $\etass_i M_i(\Ass_i)+\etass_j M_j(\Ass_j)<\etass_i M_i(\Ass_i+\Delta)+\etass_j M_j(\Ass_j-\Delta)$. Thus, through replacing $\Ass_i$ and $\Ass_j$ by $\Ass_i+\Delta$ and $\Ass_j-\Delta$ respectively, we find a better solution to $\mathbf{P}$, leading to a contradiction.
\end{proof}

\subsection{Proof of Lemma \ref{lemma_order1}}\label{appendix-lemma_order1}
\begin{proof}
Suppose $\exists i<j$ such that $M_i(\As_i)>M_j(\As_j)$. This implies that $\As_i>\As_j$. (Otherwise, suppose $\As_i\leq\As_j$, then we have $M_i(\As_i)\leq M_i(\As_j)<M_j(\As_j)$, leading to a contradiction.) A corresponding diagram is shown in Fig.~\ref{figure1}.

\textbf{Case 1:} $\etas_i\leq \etas_j$.

Let $\Asc_j=\As_i$ and $\Asc_i=\As_j$, then we have
\begin{align}
\nonumber&[\etas_iM_i(\Asc_i)+\etas_jM_j(\Asc_j)]-[\etas_iM_i(\As_i)+\etas_jM_j(\As_j)]\\
\nonumber=&[\etas_iM_i(\As_j)+\etas_jM_j(\As_i)]-[\etas_iM_i(\As_i)+\etas_jM_j(\As_j)]\\
\nonumber=&\etas_j[M_j(\As_i)-M_j(\As_j)]+\etas_i[M_i(\As_j)-M_i(\As_i)]\\
\label{formula_lemma11}\geq&\etas_i[M_j(\As_i)-M_j(\As_j)+M_i(\As_j)-M_i(\As_i)]>0,
\end{align}
where (\ref{formula_lemma11}) is due to property (M-4).

As a consequence, if $\As_i$ and $\As_j$ are replaced by $\Asc_i$ and $\Asc_j$ respectively, we obtain a larger $\mathbf{F}$, leading to a contradiction.

\textbf{Case 2:} $\etas_i> \etas_j$.

Let $\Asc_j=\As_i$, $\Asc_i=\As_j$, $\etasc_j=\etas_i$, and $\etasc_i=\etas_j$. (Note that because $\eta_{\min,i}\leq\eta_{\min,j}$ and  $\eta_{\max,i}\leq\eta_{\max,j}$,  $\etasc_j$ and $\etasc_i$ are guaranteed to be in the feasible region.)
\begin{align}
\nonumber&[\etasc_iM_i(\Asc_i)+\etasc_jM_j(\Asc_j)]-[\etas_iM_i(\As_i)+\etas_jM_j(\As_j)]\\
\nonumber=&[\etas_jM_i(\As_j)+\etas_iM_j(\As_i)]-[\etas_iM_i(\As_i)+\etas_jM_j(\As_j)]\\
\nonumber=&\etas_i[M_j(\As_i)-M_i(\As_i)]+\etas_j[M_i(\As_j)-M_j(\As_j)]\\
>&\etas_j[M_j(\As_i)-M_i(\As_i)+M_i(\As_j)- M_j(\As_j)]>0.
\end{align}
Thus, if $\As_i$, $\As_j$, $\etas_i$, and $\etas_j$ are replaced by $\Asc_i$, $\Asc_j$, $\etasc_i$, and $\etasc_j$ respectively, we can find a larger $\mathbf{F}$, leading to a contradiction.
\end{proof}

\begin{figure}[tbp]
\centering  \hspace{0pt}
\includegraphics[scale=0.5]{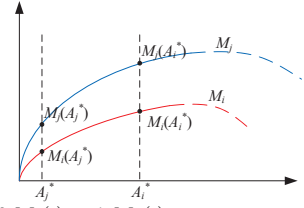}
\vspace{-0.4 cm}
\caption{Diagram of $M_i(\cdot)$ and $M_j(\cdot)$. }
\label{figure1}
\end{figure}

\subsection{Useful Properties of $\Mb_k(A_k)$ } \label{appendix-Mbk}
%Here we present some important properties of $\Mb_k(A_k)$.
\begin{description}
  \item [(M-1')] $\Mb_k(A_k)$ is a decreasing convex function.
  \item [(M-2')] If $\lambda_i<\lambda_j$, $\Mb_i(a_i)<\Mb_j(a_j)$.
  \item [(M-3')] If $\lambda_i<\lambda_j$, then $\Mb_j(A)-\Mb_i(A)$ is a strictly decreasing function. %(i.e., $\Mb_j(A)-\Mb_i(A)>\Mb_j(A')-\Mb_i(A')$ if $A<A'$).
  \item [(M-4')] $\Mb_k(A)-\Mb_k(A+D)>\Mb_k(A')-\Mb_k(A'+D)$, for any $A'>A\geq a_k$ and $D>0$.
  \item [(M-5')] If $\lambda_i<\lambda_j$, then $\Mb_j(a_j)-\Mb_j(a_j+D)>\Mb_i(a_i)-\Mb_i(a_i+D)$, for any $D>0$.
  \item [(M-6')] If $\lambda_i<\lambda_j$, then $\Mb_j(a_j)-\Mb_j(a_j+D)>\Mb_i(A')-\Mb_i(A'+D)$, for any $D>0$ and $A'>a_i$  (combining (M-4') and (M-5')).

\end{description}

\subsection{Proof of Lemma \ref{lemma_order2}} \label{appendix-lemma_order2}

\begin{proof}
Suppose that $\exists i<j$ (i.e., $\lambda_i<\lambda_j$) such that  $\Mb_i(\Asb_i)>\Mb_j(\Asb_j)$, which also implies that $a_i\leq\Asb_i<\Asb_j$. The corresponding diagrams are shown in Fig.~\ref{figure2}.

\textbf{Case 1: $\etasb_i\leq\etasb_j$.}

\textbf{Case 1.1:  $\Asb_i\geq a_j$.}

Let $\Asc_j=\Asb_i$ and $\Asc_i=\Asb_j$, then we have
\begin{align}
\nonumber&[\etasb_i\Mb_i(\Asc_i)+\etasb_j\Mb_j(\Asc_j)]-[\etasb_i\Mb_i(\Asb_i)+\etasb_j\Mb_j(\Asb_j)]\\
\nonumber=&[\etasb_i\Mb_i(\Asb_j)+\etasb_j\Mb_j(\Asb_i)]-[\etasb_i\Mb_i(\Asb_i)+\etasb_j\Mb_j(\Asb_j)]\\
\nonumber=&\etasb_j[\Mb_j(\Asb_i)-\Mb_j(\Asb_j)]+\etasb_i[\Mb_i(\Asb_j)-\Mb_i(\Asb_i)]\\
\label{formula_lemma21}\geq&\etasb_i[\Mb_j(\Asb_i)-\Mb_j(\Asb_j)+\Mb_i(\Asb_j)-\Mb_i(\Asb_i)]>0,
\end{align}
where (\ref{formula_lemma21}) is due to property (M-3').

Thus, if $\Asb_i$ and $\Asb_j$ are replaced by $\Asc_i$ and $\Asc_j$ respectively, we obtain a larger $\mathbf{F}'$, leading to a contradiction.

\textbf{Case 1.2:  $\Asb_i< a_j$.}

Let $\Asc_j=a_j$, $D=\Asb_j-a_j$ and $\Asc_i=\Asb_i+D$, then we have
\begin{align}
\nonumber&[\etasb_i\Mb_i(\Asc_i)+\etasb_j\Mb_j(\Asc_j)]-[\etasb_i\Mb_i(\Asb_i)+\etasb_j\Mb_j(\Asb_j)]\\
\nonumber=&[\etasb_i\Mb_i(\Asb_i+D)+\etasb_j\Mb_j(a_j)]-[\etasb_i\Mb_i(\Asb_i)+\etasb_j\Mb_j(\Asb_j)]\\
\nonumber=&\etasb_j[\Mb_j(a_j)-\Mb_j(\Asb_j)]+\etasb_i[\Mb_i(\Asb_i+D)-\Mb_i(\Asb_i)]\\
\label{formula_lemma22}\geq&\etasb_i[\Mb_j(a_j)-\Mb_j(\Asb_j)+\Mb_i(\Asb_i+D)-\Mb_i(\Asb_i)]>0,
\end{align}
where (\ref{formula_lemma22}) is due to property (M-6').

Thus, if $\Asb_i$ and $\Asb_j$ are replaced by $\Asc_i$ and $\Asc_j$ respectively, we obtain a larger $\mathbf{F}'$, leading to a contradiction.

\textbf{Case 2: $\etasb_i>\etasb_j$.}
\textbf{Case 2.1: $\Asb_i\geq a_j$.} \textbf{Case 2.2: $\Asb_i<a_j$.} The proof for these cases can be found in our technical report \cite{OurReport}.
\end{proof}

%\begin{figure}[tbp]
%\centering  \hspace{0pt}
%\includegraphics[scale=0.6]{figure2.eps}
%\vspace{-0.5 cm}
%\caption{Diagrams of $\Mb_i(\As_i)$ and $\Mb_j(\As_j)$. (a) corresponds to Case 1.1 and 2.1, and (b) corresponds to Case 1.2 and 2.2. }
%\label{figure2}
%\end{figure}

\begin{figure} \centering
\addtolength{\subfigcapskip}{-0.2cm}
\subfigure[Case 1.1 and 2.1.] { \label{figure2_1}
\includegraphics[scale=0.5]{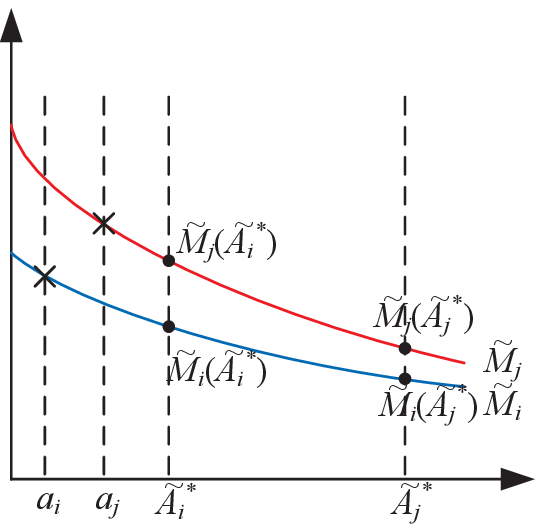}
}
\subfigure[Case 1.2 and 2.2.] { \label{figure2_2}
\includegraphics[scale=0.5]{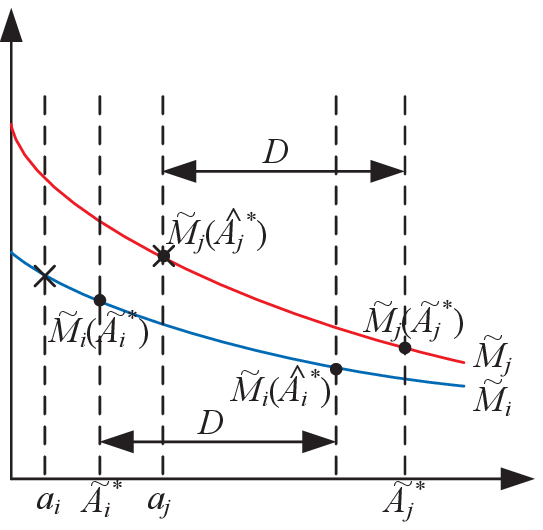}
}
\vspace{-0.3 cm}
\caption{Diagrams of $\Mb_i(\cdot)$ and $\Mb_j(\cdot)$. }
\label{figure2}
\end{figure}

\subsection{Proof of Theorem \ref{theorem_order3}} \label{appendix-theorem_order3}
\begin{proof}
Suppose $\exists k\geq 2$ such that $\Asb_k>a_k$. Let $l=1$, $\Asc_k=a_k$, $D=\Asb_k-a_k$, and $\Asc_l=\Asb_l+D$. %Note that we have $\etasb_k\geq \etasb_l$ through Lemma \ref{lemma_order2}.
Similar to the proof of Lemma \ref{lemma_order2}, we can show that if we replace $\Asb_l$ and $\Asb_k$ by $\Asc_l$ and $\Asc_k$ respectively, we find a better solution to Problem $\mathbf{P2A}$, which leads to a contradiction. See \cite{OurReport} for details.

%Thus,
%\begin{align}
%\nonumber&[\etasb_k\Mb_k(\Asc_k)+\etasb_l\Mb_l(\Asc_l)]-[\etasb_k\Mb_k(\Asb_k)+\etasb_l\Mb_l(\Asb_l)]\\
%\nonumber=&[\etasb_k\Mb_k(a_k)+\etasb_l\Mb_l(\Asb_l+D)]-[\etasb_k\Mb_k(\Asb_k)+\etasb_l\Mb_l(\Asb_l)]\\
%\nonumber=&\etasb_k[\Mb_k(a_k)-\Mb_k(\Asb_k)]+\etasb_l[\Mb_l(\Asb_l+D)-\Mb_l(\Asb_l)]\\
%\nonumber\geq&\etasb_l[\Mb_k(a_k)-\Mb_k(\Asb_k)+\Mb_l(\Asb_l+D)-\Mb_l(\Asb_l)]\\
%>&0.
%\end{align}
%As a consequence, if we replace $\Asb_l$ and $\Asb_k$ by $\Asc_l$ and $\Asc_k$ respectively, we find a better solution to Problem $\mathbf{P2A}$, which leads to a contradiction.
\end{proof}

\bibliographystyle{IEEEtran}
\bibliography{baoweiSG}

\end{document}